\documentclass[submission,publicdomain]{eptcs}

\makeatletter
\RequirePackage{amsmath}
\RequirePackage{amssymb}
\RequirePackage{amsthm}

\RequirePackage{microtype}

\RequirePackage{relsize}

\RequirePackage{xspace}
\RequirePackage{enumerate}

\newdimen\inferLineSkip		\inferLineSkip=2pt
\newdimen\inferLabelSkip	\inferLabelSkip=5pt
\def\inferTabSkip{\quad}

\newdimen\inferRuleThickness
\newdimen\@LeftOffset	
\newdimen\@RightOffset	
\newdimen\@SavedLeftOffset	

\newdimen\UpperWidth
\newdimen\LowerWidth
\newdimen\LowerHeight
\newdimen\UpperLeftOffset
\newdimen\UpperRightOffset
\newdimen\UpperCenter
\newdimen\LowerCenter
\newdimen\UpperAdjust
\newdimen\RuleAdjust
\newdimen\LowerAdjust
\newdimen\RuleWidth
\newdimen\HLabelAdjust
\newdimen\VLabelAdjust
\newdimen\WidthAdjust

\newbox\@UpperPart
\newbox\@LowerPart
\newbox\@LabelPart
\newbox\ResultBox

\newif\if@inferRule	
\newif\if@DoubleRule	
\newif\if@ReturnLeftOffset	

\def\infer@post@box#1{#1}
\def\inferPostLabel#1{#1}

\def\DeduceSym{\vtop{\baselineskip4\p@ \lineskiplimit\z@
    \vbox{\hbox{.}\hbox{.}\hbox{.}}\hbox{.}}}

\def\@IFnextchar#1#2#3{%
  \let\reserved@e=#1\def\reserved@a{#2}\def\reserved@b{#3}\futurelet
    \reserved@c\@IFnch}
\def\@IFnch{\ifx \reserved@c \@sptoken \let\reserved@d\@xifnch
      \else \ifx \reserved@c \reserved@e\let\reserved@d\reserved@a\else
          \let\reserved@d\reserved@b\fi
      \fi \reserved@d}

\def\@ifEmpty#1#2#3{\def\@tempa{\@empty}\def\@tempb{#1}\relax
	\ifx \@tempa \@tempb #2\else #3\fi }

\def\infer{\@IFnextchar *{\@inferSteps}{\relax
	\@IFnextchar ={\@inferDoubleRule}{\@inferOneStep}}}

\def\linfer{\def\infer@post@box##1{\lower\LowerHeight##1}\infer}

\def\@inferOneStep{\@inferRuletrue \@DoubleRulefalse
	\@IFnextchar [{\@infer}{\@infer[\@empty]}}

\def\@inferDoubleRule={\@inferRuletrue \@DoubleRuletrue
	\@IFnextchar [{\@infer}{\@infer[\@empty]}}

\def\@inferSteps*{\@IFnextchar [{\@@inferSteps}{\@@inferSteps[\@empty]}}

\def\@@inferSteps[#1]{\@deduce{#1}[\DeduceSym]}

\def\deduce{\@IFnextchar [{\@deduce{\@empty}}
	{\@inferRulefalse \@infer[\@empty]}}

\def\@deduce#1[#2]#3#4{\@inferRulefalse
	\@infer[\@empty]{#3}{\@infer[{#1}]{#2}{#4}}}

\def\@infer[#1]#2#3{\relax
	\if@ReturnLeftOffset \else \@SavedLeftOffset=\@LeftOffset \fi
	\setbox\@LabelPart=\hbox{\inferPostLabel{#1}}\relax
	\setbox\@LowerPart=\hbox{$#2$}\relax
	\global\@LeftOffset=0pt
	\setbox\@UpperPart=\vbox{\tabskip=0pt\let\infer@post@box=\relax\halign{\relax
		\global\@RightOffset=0pt \@ReturnLeftOffsettrue $##$&&
		\inferTabSkip
		\global\@RightOffset=0pt \@ReturnLeftOffsetfalse $##$\cr
		#3\cr}}\relax
	\UpperLeftOffset=\@LeftOffset
	\UpperRightOffset=\@RightOffset
	\LowerWidth=\wd\@LowerPart
	\LowerHeight=\ht\@LowerPart
	\LowerCenter=0.5\LowerWidth
	\UpperWidth=\wd\@UpperPart \advance\UpperWidth by -\UpperLeftOffset
	\advance\UpperWidth by -\UpperRightOffset
	\UpperCenter=\UpperLeftOffset
	\advance\UpperCenter by 0.5\UpperWidth
	\ifdim \UpperWidth > \LowerWidth
	\UpperAdjust=0pt
	\RuleAdjust=\UpperLeftOffset
	\LowerAdjust=\UpperCenter \advance\LowerAdjust by -\LowerCenter
	\RuleWidth=\UpperWidth
	\global\@LeftOffset=\LowerAdjust
	\else	
	\ifdim \UpperCenter > \LowerCenter
	\UpperAdjust=0pt
	\RuleAdjust=\UpperCenter \advance\RuleAdjust by -\LowerCenter
	\LowerAdjust=\RuleAdjust
	\RuleWidth=\LowerWidth
	\global\@LeftOffset=\LowerAdjust
	\else	
	\UpperAdjust=\LowerCenter \advance\UpperAdjust by -\UpperCenter
	\RuleAdjust=0pt
	\LowerAdjust=0pt
	\RuleWidth=\LowerWidth
	\global\@LeftOffset=0pt
	\fi\fi
	\if@inferRule
	\setbox\ResultBox=\vbox{
		\moveright \UpperAdjust \box\@UpperPart
		\nointerlineskip \kern\inferLineSkip
		\if@DoubleRule
		\moveright \RuleAdjust \vbox{\hrule height \inferRuleThickness width\RuleWidth
			\kern 1pt\hrule height \inferRuleThickness width\RuleWidth}\relax
		\else
		\moveright \RuleAdjust \vbox{\hrule height \inferRuleThickness width\RuleWidth}\relax
		\fi
		\nointerlineskip \kern\inferLineSkip
		\moveright \LowerAdjust \box\@LowerPart }\relax
	\@ifEmpty{#1}{}{\relax
	\HLabelAdjust=\wd\ResultBox	\advance\HLabelAdjust by -\RuleAdjust
	\advance\HLabelAdjust by -\RuleWidth
	\WidthAdjust=\HLabelAdjust
	\advance\WidthAdjust by -\inferLabelSkip
	\advance\WidthAdjust by -\wd\@LabelPart
	\ifdim \WidthAdjust < 0pt \WidthAdjust=0pt \fi
	\VLabelAdjust=\dp\@LabelPart
	\advance\VLabelAdjust by -\ht\@LabelPart
	\VLabelAdjust=0.5\VLabelAdjust	\advance\VLabelAdjust by \LowerHeight
	\advance\VLabelAdjust by \inferLineSkip
	\setbox\ResultBox=\hbox{\box\ResultBox
		\kern -\HLabelAdjust \kern\inferLabelSkip
		\raise\VLabelAdjust \box\@LabelPart \kern\WidthAdjust}\relax
	}\relax 
	\else 
	\setbox\ResultBox=\vbox{
		\moveright \UpperAdjust \box\@UpperPart
		\nointerlineskip \kern\inferLineSkip
		\moveright \LowerAdjust \hbox{\unhbox\@LowerPart
			\@ifEmpty{#1}{}{\relax
			\kern\inferLabelSkip \unhbox\@LabelPart}}}\relax
	\fi
	\global\@RightOffset=\wd\ResultBox
	\global\advance\@RightOffset by -\@LeftOffset
	\global\advance\@RightOffset by -\LowerWidth
	\if@ReturnLeftOffset \else \global\@LeftOffset=\@SavedLeftOffset \fi
	\infer@post@box{\box\ResultBox}
}

\IfFileExists{upgreek}{\RequirePackage[Symbol]{upgreek}}{}

\RequirePackage{graphicx}
\RequirePackage{xcolor}

\RequirePackage{tikz}

\usetikzlibrary{positioning}
\usetikzlibrary{matrix}
\usetikzlibrary{calc}
\usetikzlibrary{backgrounds}

\allowdisplaybreaks[4]

\definecolor{highlight}{HTML}{B22222}

\newcommand\ie{\textit{i.\@ e.\@},\xspace}

\newcommand\etc{\textit{etc}\@ifnextchar.\dot@{.\@}}
\def\dot@.{\@.\xspace}

\newcommand\providetheorem[2][plain]{%
\@ifundefined{#2}{\theoremstyle{#1}\newtheorem{#2}}{\ignorenewtheorem}}
\newcommand\ignorenewtheorem[2][]{}

\providetheorem[plain]{theorem}{Theorem}
\providetheorem[lemma]{lemma}[theorem]{Lemma}
\providetheorem[plain]{corollary}[theorem]{Corollary}
\providetheorem[plain]{proposition}[theorem]{Proposition}
\providetheorem[plain]{fact}[theorem]{Fact}
\providetheorem[plain]{remark}[theorem]{Remark}
\providetheorem[plain]{example}[theorem]{Example}
\providetheorem[plain]{conjecture}[theorem]{Conjecture}
\providetheorem[definition]{notation}[theorem]{Notation}
\providetheorem[definition]{convention}[theorem]{Convention}
\providetheorem[definition]{definition}[theorem]{Definition}

\newcommand\defin{\hfill$\lrcorner$}

\newcommand\GOR {\mid}

\let\vert@vert=\|
\let\vert@=|
\def\vert@minus-{\mathop{\mathstrut{\vdash}}\nolimits}%
\def\vert@eq={\mathbin{\mathstrut{\Vdash\,}}}%
\def\vert@gt>{\mathbin{\gg}}%
\def\vert@vert!{\vert@}%
\catcode`|=\active
\def|{%
\@ifnextchar-\vert@minus{%
\@ifnextchar=\vert@eq{%
\@ifnextchar>\vert@gt{%
\@ifnextchar!\vert@vert{%
\GOR%
}}}}}

\newcommand\set[1]{\left\{#1\right\}}
\newcommand\str[1]{\left\langle#1\right\rangle}

\providecommand\amp{\ensuremath{\text{\sffamily\&}}}
\providecommand\invamp{\rotatebox[origin=c]{180}{\amp}}

\newcommand\TENS {\mathbin{\mathstrut{\pmb\otimes}}}
\newcommand\ONE  {\mathsf{\mathstrut1}}
\newcommand\PLUS {\mathbin{\mathstrut{\pmb\oplus}}}
\newcommand\ZERO {\mathsf{\mathstrut0}}
\newcommand\PAR  {\mathbin{\mathstrut{\invamp}}}
\newcommand\BOT  {{\mathstrut\pmb\bot}}
\newcommand\WITH {\mathbin{\mathstrut\amp}}
\newcommand\TOP  {{\mathstrut\pmb\top}}

\newcommand\BANG {\mathop{\mathstrut\text{\sffamily !}}\nolimits}
\newcommand\QM   {\mathop{\mathstrut\text{\sffamily ?}}\nolimits}
\newcommand\NEG  {\mathop{\mathstrut\lnot}}

\def\ALL#1.{\mathstrut\forall #1.\,}
\def\EX#1.{\mathstrut\exists #1.\,}

\def\MUS{\upmu}
\def\NUS{\upnu}
\def\MU#1.{\MUS\mkern 1mu{#1}.\,}
\def\NU#1.{\NUS\mkern 1mu{#1}.\,}

\def\LAM#1.{\uplambda #1.\,}

\newcommand\LOC{\Lambda}

\newcommand\HOLE {\Box}

\def\<{\mkern -2mu\left\langle}
\def\>{\right\rangle\mkern -2mu}

\providecommand\Upgamma{\Gamma}
\def\G{\Upgamma}
\providecommand\Updelta{\Delta}
\def\D{\Updelta}
\providecommand\Upomega{\Omega}
\def\W{\Upomega}
\providecommand\Upsigma{\Sigma}
\def\Si{\Upsigma}

\def\foc#1{\left[{#1}\right]}
\def\focrn#1{\ensuremath{\foc{\rn{#1}}}\xspace}

\colorlet{hil}{black!15!white}

\newcommand\ssh[3][hil]{%
\bgroup%
\def\HOLE{{\fboxsep 0pt\colorbox{#1}{$#3$}}}%
#2%
\egroup}

\newcommand\proofsystem[1]{\ensuremath{\text{\smaller\rmfamily\slshape #1}}\xspace}

\DeclareRobustCommand\muMAIS{\@ifstar{{\mu@MAIS}\ensuremath{^*}}{\mu@MAIS}}
\def\mu@MAIS{\proofsystem{$\mu$MAIS}}

\newcommand\MLL{\proofsystem{MLL}}
\newcommand\MELL{\proofsystem{MELL}}
\newcommand\MALL{\proofsystem{MALL}}
\newcommand\MAELL{\proofsystem{MAELL}}
\newcommand\SEL{\proofsystem{SEL}}
\newcommand\SELL{\proofsystem{SELL}}
\newcommand\MSEL{\proofsystem{MSEL}}
\newcommand\TWORM{\proofsystem{2RM}}

\renewcommand\inferTabSkip{\hspace{1em minus .5em}}
\def\inferPostLabel#1{\hbox{\smaller[2]$#1$}}
\newcommand\rn[1]{\ensuremath{\text{\sffamily\upshape #1}}\xspace}

\newcommand\xinfer[1][]{\def\xinfer@label{#1}\xinfer@}
\newcommand{\xinfer@}[3][\cx]{\linfer[\xinfer@label]{\ssx[#1]{#2}}{\ssx[#1]{#3}}}

\let\deriv@main=\infer
\newcommand\deriv[1]{\deriv@#1;\deriv@end}
\def\deriv@#1;#2\deriv@end{%
\@ifempty{#2}{#1}{%
\@ifnextchar[\deriv@named\deriv@anon#1\deriv@t{\let\deriv@main=\infer\deriv@#2\deriv@end}%
}}
\def\deriv@named[#1]#2\deriv@t#3{\deriv@main[#1]{#2}{#3}}
\def\deriv@anon#1\deriv@t#2{\deriv@main{#1}{#2}}
\newcommand\lderiv[1]{\let\deriv@main=\linfer\deriv@#1,\deriv@end}

\newcommand\Der[2][]{\ensuremath{\Der@#1\Der@pc#2\Der@end}}
\def\Der@#1\Der@pc#2,#3\Der@end{%
  #2 \smash{\xrightarrow{\ #1\ }} #3
}

\renewcommand\k[1]{\text{\ttfamily #1}}
\newcommand\ka{\k{a}}
\newcommand\kb{\k{b}}
\newcommand\kh{\k{h}}

\newcommand\CONF{\mathcal{C}}
\newcommand\TRANS[1][]{{\xrightarrow{\ #1\ }}}
\newcommand\kl[1]{\hbox to 1cm{\scriptsize\hfill$\k{#1}$\hfill}}

\newcommand\TWOSI{\Xi}
\newcommand\UNB{\infty}

\newcommand\kra{\k{ra}}
\newcommand\krb{\k{rb}}

\newcommand\ENC[1]{\mathcal E(#1)}
\makeatother

\usepackage{breakurl}

\providecommand\event{Linearity 2014}

\title{Undecidability of Multiplicative Subexponential Logic}

\author{%
  Kaustuv Chaudhuri%
  \institute{INRIA, France}%
  \email{kaustuv.chaudhuri@inria.fr}
}

\def\titlerunning{%
  Undecidability of Multiplicative Subexponential Logic}
\def\authorrunning{K. Chaudhuri}

\begin{document}

\maketitle

\begin{abstract}
  Subexponential logic is a variant of linear logic with a family of
  exponential connectives---called \emph{subexponentials}---that are
  indexed and arranged in a pre-order.
  Each subexponential has or lacks associated structural properties of
  weakening and contraction.
  We show that classical propositional multiplicative linear logic
  extended with one unrestricted and two incomparable linear
  subexponentials can encode the halting problem for two register
  Minsky machines, and is hence undecidable.
\end{abstract}

\section{Introduction}
\label{sec:introduction}

The decision problem for classical propositional multiplicative
exponential linear logic (\MELL), consisting of formulas constructed
from propositional atoms using the connectives $\set{\TENS,
  \ONE, \PAR, \BOT, \BANG, \QM}$, is perhaps the longest standing open
problem in linear logic.
\MELL is bounded below by the purely multiplicative fragment (\MLL),
which is decidable even in the presence of first-order quantification,
and above by \MELL with additive connectives (\MAELL), which is
undecidable even for the propositional fragment~\cite{lincoln90apal}.
This paper tries to make the undecidable upper bound a bit tighter by
considering the question of the decision problem for a family of
propositional multiplicative \emph{subexponential} logics
(\MSEL)~\cite{nigam09phd,nigam09ppdp-alt}, each of which consists of
formulas constructed from propositional atoms using the (potentially
infinite) set of connectives $\set{\TENS,\ONE,\PAR,\BOT} \cup \
\bigcup_{u \in \Si} \set{\BANG^u, \QM^u}$, where $\Si$ is a
pre-ordered set of subexponential \emph{labels}, called a
\emph{subexponential signature}, that is a parameter of the family of
logics.
In particular, we show that a particular \MSEL with a subexponential
signature consisting of exactly three labels can encode a two register
Minsky machine (\TWORM), which is Turing-equivalent.
This is the same strategy used in~\cite{lincoln90apal} to show the
undecidability of \MAELL, but the encoding in \MSEL is
different---simpler---for the branching instructions, and shows that
additive behaviour is not essential to implement branching.
We use the classical dialect of linear logic to show these results.
The intuitionistic dialect has the same decision problem because it is
possible to faithfully encode (\ie linearly simulate the sequent
proofs of) the classical dialect in the intuitionistic dialect without
changing the signature~\cite{local:chaudhuri10csl}.

This short note is organized as follows: in
section~\ref{sec:background} we sketch the one-sided sequent
formulation of \MSEL and recall the definition of a \TWORM.
In section~\ref{sec:encoding} we encode the transition system of a
\TWORM in a \MSEL with a particular signature.
In section~\ref{sec:adequacy} we argue that the encoding is
\emph{adequate}, \ie that the halting problem for a \TWORM is reduced
to the proof search problem for this \MSEL-encoding, by appealing to a
focused sequent calculus for \MSEL.
The final section~\ref{sec:perspectives} discusses some of the
ramifications of this result.

\section{Background}
\label{sec:background}

\subsection{Propositional Subexponential Logic}
\label{sec:subexponential-logic}

Let us quickly recall propositional subexponential logic (\SEL) and
its associated sequent calculus proof system.
This logic is sometimes called subexponential \emph{linear} logic
(\SELL), but since it is possible for the subexponentials to have
linear semantics it is redundant to include both adjectives.
Formulas of \SEL ($A, B, \dotsc$) are built from \emph{atomic
  formulas} ($a, b, \dotsc$) according to the following grammar:

\begin{tikzpicture}[node distance=1ex]
  \node [matrix of math nodes] (gr) {
    A, B, \dotsc & ::= &
    \node(a){a}; & \GOR & \node(tens){A \TENS B}; & \GOR & \ONE & \GOR &
    \node(plus){A \PLUS B}; & \GOR & \ZERO &
    \GOR & \node(bang){\BANG^u A}; \\
    & \node[right] {\GOR} ; &
    \node(a'){\NEG a}; & \GOR &
    \node(par){A \PAR B}; & \GOR & \node(bot){\BOT}; & \GOR &
    \node(with){A \WITH B}; & \GOR & \node(top){\TOP}; &
    \GOR & \node(qm){\QM^u A}; \\
  } ;
  \begin{scope}[on background layer]
    \fill[rounded corners,color=green!5!white]
       ($(a.north west)-(.2,0)$) rectangle ($(a'.south east)-(0,.5)$) ;
    \node at ($(a'.south west)!.5!(a'.south east)-(0.05,.2)$){
      \tiny\scshape atomic
    } ;
    \fill[rounded corners,color=blue!5!white]
       (tens.north west) rectangle ($(bot.south east)-(0,.5)$) ;
    \node at ($(par.south west)!.5!(bot.south east)-(0,.2)$) {
      \tiny\scshape multiplicative
    } ;
    \fill [rounded corners,color=red!5!white]
       (plus.north west) rectangle ($(top.south east)-(0,.5)$) ;
    \node at ($(with.south west)!.5!(top.south east)-(0,.2)$) {
      \tiny\scshape additive
    } ;
    \fill [rounded corners,color=orange!15!white]
       (bang.north west) rectangle ($(qm.south east)+(1,-.5)$) ;
    \node at ($(qm.south west)!.5!(qm.south east)+(.5,-.2)$) {
      \tiny \scshape subexponential
    } ;
  \end{scope}
\end{tikzpicture}

\noindent%
Each column in the grammar above is a De Morgan dual pair.
A \emph{positive formula} (depicted with $P$ or $Q$ when relevant) is
a formula belonging to the first line of the grammar, and a
\emph{negative formula} (depicted with $N$ or $M$) is a formula
belonging to the second line.
The \emph{labels} ($u, v, \dotsc$) on the subexponential connectives
$\BANG^u$ and $\QM^u$ belong to a \emph{subexponential signature}
defined below.
The additive fragment of this syntax is just used in this section for
illustration; we will not be using the additives in our encodings.
The fragment without the additives will be called \emph{multiplicative
  subexponential logic} (\MSEL).

\begin{definition}
  A \emph{subexponential signature} $\Si$ is a structure $\str{\LOC,
    U, \le}$ where:
  \begin{itemize}
  \item $\LOC$ is a countable set of \emph{labels};
  \item $U \subseteq \LOC$, called the \emph{unbounded labels}; and
  \item ${\le} \subseteq \LOC \times \LOC$ is a pre-order on $\LOC$---
    \ie it is reflexive and transitive---and $\le$-upwardly closed
    with respect to $U$, \ie for any $u,v\in \LOC$, if $u \in U$ and
    $u \le v$, then $v \in U$. \defin
  \end{itemize}
\end{definition}

\noindent%
We will assume an ambient signature $\Si$ unless we need to
disambiguate particular instances of \MSEL, in which case we will
use $\Si$ in subscripts.
For instance, $\MSEL_\Si$ is a particular instance of \MSEL for $\Si$.

\begin{figure}
  \centering \small
  \begin{gather*}
    \linfer[\rn{init}]{|- a, \NEG a}{}
    \quad
    \linfer[\TENS]{
      |- \G, \D, A \TENS B
    }{
      |- \G, A & |- \D, B
    }
    \quad
    \linfer[\ONE]{
      |- \ONE
    }{}
    \quad
    \linfer[\PLUS_1]{
      |- \G, A \PLUS B
    }{
      |- \G, A
    }
    \quad
    \linfer[\PLUS_2]{
      |- \G, A \PLUS B
    }{
      |- \G, B
    }
    \quad
    \infer[\text{no rule for $\ZERO$}]{}{}
    \\[1ex]
    \linfer[\PAR]{
      |- \G, A \PAR B
    }{
      |- \G, A, B
    }
    \quad
    \linfer[\BOT]{
      |- \G, \BOT
    }{
      |- \G
    }
    \quad
    \linfer[\WITH]{
      |- \G, A \WITH B
    }{
      |- \G, A & |- \G, B
    }
    \quad
    \linfer[\TOP]{
      |- \G, \TOP
    }{}
    \quad
    \linfer[\QM]{
      |- \G, \QM^u A
    }{
      |- \G, A
    }
    \\[1ex]
    \linfer[\BANG]{
      |-_\Si\ \QM^{\vec v} \vec A, \BANG^u C
    }{
      (u \le_\Si \vec v)
      &
      |-_\Si\ \QM^{\vec v} \vec A, C
    }
    \qquad
    \linfer[\rn{weak}]{
      |-_\Si\ \G, \QM^u A
    }{
      (u \in U_\Si)
      &
      |-_\Si\ \G
    }
    \qquad
    \linfer[\rn{contr}]{
      |-_\Si\ \G, \QM^u A
    }{
      (u \in U_\Si)
      &
      |-_\Si\ \G, \QM^u A, \QM^u A
    }
  \end{gather*}
  \caption{%
    Inference rules for a cut-free one-sided sequent calculus formulation of
    \SEL.
    Only the rules on the last line are sensitive to the signature. }
  \label{fig:sel-rules}
\end{figure}

The true formulas of \MSEL are derived from a \emph{sequent calculus}
proof system consisting of sequents of the form $|- A_1, \dotsc, A_n$
(with $n > 0$) and abbreviated as $|- \G$.
The \emph{contexts} ($\G, \D, \ldots$) are multi-sets of formulas of
\SEL, and $\G,\D$ and $\G, A$ stand as usual for the multi-set union
of $\G$ with $\D$ and $\set{A}$, respectively.
The inference rules for \SEL sequents are displayed in
figure~\ref{fig:sel-rules}.
Most of the rules are shared between \SEL and linear logic and will
not be elaborated upon here.
The differences are with the subexponentials, for which we use the
following definition.

\begin{definition}
  For any $n \in N$ and lists $\vec u = [u_1, \dotsc, u_n]$ and $\vec
  A = [A_1, \dotsc, A_n]$, we write $\QM^{\vec u} \vec A$ to stand for
  the context $\QM^{u_1} A_1, \dotsc, \QM^{u_n} A_n$.
  For $\vec v = [v_1, \dotsc, v_n]$, we write $u \le \vec v$ to mean
  that $u \le v_1$, \ldots, and $u \le v_n$. \defin
\end{definition}

The rule for $\BANG$, sometimes called \emph{promotion}, has a side
condition that checks that the label of the principal formula is less
than the labels of all the other formulas in the context.
This rule cannot be used if there are non-$\QM$-formulas in the
context, nor if the labels of some of the $\QM$-formulas are strictly
smaller or incomparable with that of the principal $\BANG$-formula.
Both these properties will be used in the encoding in the next section.
The structural rules of weakening and contraction apply to those
principal $\QM$-formulas with unbounded labels.

%
%

\subsection{Two Register Minsky Machines}
\label{sec:two-register-minsky}

Like Turing machines, Minsky register machines have a finite state
diagram and transitions that can perform I/O on some unbounded storage
device, in this case a bank of registers that can store arbitrary
natural numbers.
We shall limit ourselves to machines with two registers (\TWORM) $\ka$
and $\kb$, which are sufficient to encode Turing machines.

\begin{definition}
  A \TWORM is a structure $\str{Q, *, \CONF, \TRANS{}}$ where:
  \begin{itemize}
  \item $Q$ is a non-empty finite set of \emph{states};
  \item $* \in Q$ is a distinguished \emph{halting state};
  \item $\CONF$ is a set of \emph{configurations}, each of which is a
    structure of the form $\str{q, v}$, with $q \in Q$ and $v :
    \set{\ka, \kb} \to N$, that assigns values (natural numbers) to
    the registers $\ka$ and $\kb$ in state $q$;
  \item $\TRANS{} \subseteq \CONF \times I \times \CONF$ is a
    deterministic labelled transition relation between configurations
    where the label set $I = \set{\k{halt}, \k{incra}, \k{incrb},
      \k{decra}, \k{decrb}, \k{isza}, \k{iszb}}$ (called the
    \emph{instructions}).
  \end{itemize}
  By usual convention, we write $\TRANS{}$ infix with the instruction
  atop the arrow.
  We require that every element of $\TRANS$ fits one of the following
  schemas, where in each case $q, r \in Q$ and $q \neq r$:
  \begin{gather}
    \begin{array}{rcl}
      \str{q, v}
      &\TRANS[\kl{halt}]&
      \str{*, \set{\ka:0, \kb:0}}
      \rlap{\hspace{1cm} $(\text{with }q \neq *)$}
      \\
      \str{q, \set{\ka:m, \kb:n}}
      &\TRANS[\kl{incra}]&
      \str{r, \set{\ka: m + 1, \kb:n}}
      \\
      \str{q, \set{\ka:m, \kb:n}}
      &\TRANS[\kl{incrb}]&
      \str{r, \set{\ka:m, \kb:n + 1}}
      \\
      \str{q, \set{\ka:m + 1, \kb:n}}
      &\TRANS[\kl{decra}]&
      \str{r, \set{\ka: m, \kb:n}}
      \\
      \str{q, \set{\ka:m, \kb:n + 1}}
      &\TRANS[\kl{decrb}]&
      \str{r, \set{\ka:m, \kb:n}}
      \\
      \str{q, \set{\ka:0, \kb:n}}
      &\TRANS[\kl{isza}]&
      \str{r, \set{\ka:0, \kb:n}}
      \\
      \str{q, \set{\ka:m, \kb:0}}
      &\TRANS[\kl{iszb}]&
      \str{r, \set{\ka:m, \kb:0}}
    \end{array}
    \label{eq:trans}
  \end{gather}
  For a \emph{trace} $\vec i = [i_1, \dotsc, i_n]$, we write
  $\str{q_0,v_0} \TRANS[\vec i] \str{q_n, v_n}$ if $\str{q_0, v_0}
  \TRANS[i_1] \dotsm \TRANS[i_n] \str{q_n, v_n}$.
  The \TWORM \emph{halts from} an initial configuration $\str{q_0,
    v_0}$ if there is a trace $\vec i$ such that $\str{q_0, v_0}
  \TRANS[\vec i] \str{*, \set{\ka:0, \kb:0}}$.
  (The configuration $\str{*, \set{\ka:0, \kb:0}}$ will be called the
  \emph{halting configuration}.)
  The \emph{halting problem} for a \TWORM is the decision problem of
  whether the machine halts from an initial configuration. \defin
\end{definition}

The requirement that $\TRANS$ be deterministic amounts to: $\str{q, v}
\TRANS[i] \str{q_1, v_1}$ and $\str{q, v} \TRANS[j] \str{q_2, v_2}$
imply that $i = j$, $q_1 = q_2$, and $v_1 = v_2$.
Note that a trace that does not end with a halting configuration will
not be considered to be halting, even if there is no possible
successor configuration.
It is an easy exercise to transform a given \TWORM into one where
every configuration has a successor except for the halting
configuration.

\medskip

\begin{theorem}[\cite{minsky61aom}] \label{thm:2rm-undec}%
  The halting problem for {\TWORM}s is recursively unsolvable. \qed
\end{theorem}

\section{The Encoding}
\label{sec:encoding}

For a given \TWORM, which we fix in this section, we will encode its
halting problem as the derivability of a particular \MSEL sequent that
encodes its labelled transition system and the initial configuration.
We will use the following subexponential signature in the rest of this
section.

\begin{definition}
  Let $\TWOSI$ stand for the signature $\str{\set{\UNB, \ka, \kb},
    \set{\UNB}, \le}$ where $\le$ is the reflexive-transitive closure
  of $\le_0$ defined by $\ka \le_0 \UNB$ and $\kb \le_0 \UNB$.
  \defin
\end{definition}

\begin{definition}[encoding configurations]
  For $c = \str{q, v}$, we write $\ENC{c}$ for the following
  $\MSEL_\TWOSI$ context:
  \begin{gather*}
    \underbrace{\QM^\ka \NEG \kra, \QM^\ka \NEG \kra, \dotsc,
      \QM^\ka \NEG \kra}_{\text{length } =\ v(\ka)},
    \underbrace{\QM^\kb \NEG \krb, \QM^\kb \NEG \krb, \dotsc,
      \QM^\kb \NEG \krb}_{\text{length } =\ v(\kb)},
    \NEG q
    \tag*{\defin}
  \end{gather*}
\end{definition}

\begin{definition}[encoding transitions]
  The transitions~(\ref{eq:trans}) of the \TWORM are encoded as a
  context $\Pi$ with:
  \begin{itemize}
  \item to represent $\str{q,v} \TRANS[\k{halt}] \str{*,
      \set{\ka:0,\kb:0}}$, the elements: $q \TENS \NEG \kh, \kh \TENS
    \BANG^\ka \kra \TENS \NEG \kh, \kh \TENS \BANG^\kb \krb \TENS \NEG
    \kh, \kh \TENS \BANG^\infty \ONE$ (for some $\kh \notin Q$):
  \item to represent $\str{q, \set{\ka:m, \kb:n}} \TRANS[\k{incra}]
    \str{r, \set{\ka:m + 1, \kb:n}}$, the element $q \TENS (\NEG
    r \PAR \QM^\ka \NEG \kra)$;
  \item to represent $\str{q, \set{\ka:m, \kb:n}} \TRANS[\k{incrb}]
    \str{r, \set{\ka:m, \kb:n + 1}}$, the element: $q \TENS (\NEG
    r \PAR \QM^\kb \NEG \krb)$;
  \item to represent $\str{q, \set{\ka:m + 1, \kb:n}}
    \TRANS[\k{decra}] \str{r, \set{\ka:m, \kb:n}}$, the element: $q
    \TENS \BANG^\ka \kra \TENS \NEG r$;
  \item to represent $\str{q, \set{\ka:m, \kb:n + 1}}
    \TRANS[\k{decrb}] \str{r, \set{\ka:m, \kb:n}}$, the element: $q
    \TENS \BANG^\kb \krb \TENS \NEG r$;
  \item to represent $\str{q, \set{\ka:0, \kb:n}} \TRANS[\k{isza}]
    \str{r, \set{\ka:0, \kb:n}}$, the element: $q \TENS \BANG^\kb \NEG
    r$; and
  \item to represent $\str{q, \set{\ka:m, \kb:0}} \TRANS[\k{iszb}]
    \str{r, \set{\ka:m, \kb:0}}$, the element: $q \TENS \BANG^\ka \NEG
    r$.
  \end{itemize}
  Note that $\Pi$ contains a finite number of elements. \defin
\end{definition}

\begin{definition}[encoding the halting problem]
  If $\G$ is $A_1, \dotsc, A_n$, then let $\QM^u \G$ stand for $\QM^u
  A_1, \dotsc, \QM^u A_n$.
  The encoding of the halting problem for the \TWORM from the initial
  configuration $c_0 = \str{q_0, v_0}$ is the $\MSEL_\TWOSI$ sequent
  $|- \QM^\infty \Pi, \ENC{c_0}$. \defin
\end{definition}

\begin{theorem} \label{thm:forward}%
  If the \TWORM halts from $c_0$, then $|-_\TWOSI \QM^\infty \Pi,
  \ENC{c_0}$ is derivable.
\end{theorem}

\begin{proof}
  We will show that if $c = \str{q_1, v_1} \TRANS[i] \str{q_2, v_2} =
  d$ (for some $i$), then the following $\MSEL_\TWOSI$ rule is
  derivable:

  \vspace{-2em}
  \begin{gather*}
    \infer{
      |- \QM^\infty \Pi, \ENC{c}
    }{
      |- \QM^\infty \Pi, \ENC{d}
    }
  \end{gather*}
  This is largely immediate by inspection. Here are three
  representative cases.
  \begin{itemize}
  \item The case of $i = \k{incra}$: it must be that $v_2(\ka) =
    v_1(\ka) + 1$ and $v_2(\kb) = v_1(\kb)$, so $\ENC{d} =
    \ENC{c}\setminus\set{\NEG q_1}, \NEG q_2, \QM^\ka \kra$.
    Moreover, $q_1 \TENS (\NEG q_2 \PAR \QM^\ka \NEG \kra) \in \Pi$.
    So:
    \begin{gather*}
      \infer[\rn{contr}, \QM]{
        |- \QM^\infty \Pi, \ENC{c}
      }{
      \infer[\TENS]{
        |- \QM^\infty \Pi, \ENC{c}, q_1 \TENS (\NEG q_2 \PAR \QM^\ka \NEG \kra)
      }{
        \infer[\rn{init}]{|- \NEG q_1, q_1}{} &
      \infer[\PAR]{
        |- \QM^\infty \Pi, \ENC{c}\setminus\set{\NEG q_1}, \NEG q_2 \PAR \QM^\ka \NEG \kra
      }{
        |- \QM^\infty \Pi, \ENC{c}\setminus\set{\NEG q_1}, \NEG q_2, \QM^\ka \NEG \kra
      }}}
    \end{gather*}
    The cases for \k{incrb}, \k{decra}, and \k{decrb} are similar.
  \item The case of $i = \k{isza}$: it must be that $v_2(\ka) =
    v_1(\ka) = 0$ and $v_2(\kb) = v_1(\kb)$.
    Hence, $\ENC{d} = \ENC{c}\setminus\set{\NEG q_1}, \NEG q_2$ and
    $\QM^\ka \kra \notin \ENC{c} \cup \ENC{d}$.
    Moreover, $q_1 \TENS \BANG^\kb \NEG q_2 \in \Pi$. So:
    \begin{gather*}
      \infer[\rn{contr}, \QM]{
        |- \QM^\infty \Pi, \ENC{c}
      }{
      \infer[\TENS]{
        |- \QM^\infty \Pi, \ENC{c},
           q_1 \TENS \BANG^\kb \NEG q_2
      }{
        \infer[\rn{init}]{|- \NEG q_1, q_1}{}
        &
      \infer[\BANG]{
        |- \QM^\infty \Pi, \ENC{c}\setminus\set{\NEG q_1},
           \BANG^\kb \NEG q_2
      }{
        |- \QM^\infty \Pi, \ENC{c}\setminus\set{\NEG q_1},
           \NEG q_2
      }}}
    \end{gather*}
    The instance of $\BANG$ is justified because $\kb \le \infty$ and
    $\kb \le \kb$, and there are no $\QM$-formulas labelled $\ka$ or
    non-$\QM$ formulas in the sequent.
    The case of \k{iszb} is similar.
  \item The case of $i = \k{halt}$.
    Here, we know that $q_1 \TENS \NEG \kh \in \Pi$, so:
    \begin{gather*}
      \infer[\rn{contr}, \QM]{
        |- \QM^\infty \Pi, \ENC{c}
      }{
      \infer[\TENS]{
        |- \QM^\infty \Pi, \ENC{c}, q_1 \TENS \NEG \kh
      }{
        \infer[\rn{init}]{|- \NEG q_1, q_1}{}
        &
        |- \QM^\infty \Pi, \ENC{c} \setminus\set{\NEG q_1}, \NEG \kh
      }}
    \end{gather*}
    Now, as long as there are any occurrences of $\QM^\ka \kra$ or
    $\QM^\ka \krb$ in $\ENC{c}$, we can apply one of the decrementing
    rules $\kh \TENS \BANG^\ka \kra \TENS \NEG \kh$ or $\kh \TENS
    \BANG^\kb \krb \TENS \NEG \kh \in \Pi$.
    The general case looks something like this, where $\D_\kra =
    \set{\NEG \kra, \dotsc, \NEG \kra}$ and $\D_\krb = \set{\NEG \krb,
      \dotsc, \NEG \krb}$.
    \begin{gather*}
      \infer[\rn{contr},\QM]{
        |- \QM^\infty \Pi,
        \ENC{c} \setminus\set{\smash{\NEG q_1,
          \QM^\ka \D_\kra, \QM^\kb \D_\krb, \QM^\ka \NEG \kra}},
        \QM^\ka \NEG \kra,
        \NEG \kh
      }{
      \infer[\TENS, \TENS]{
        |- \QM^\infty \Pi,
        \ENC{c} \setminus\set{\smash{\NEG q_1,
          \QM^\ka \D_\kra, \QM^\kb \D_\krb, \QM^\ka \NEG \kra}},
        \QM^\ka \NEG \kra,
        \NEG \kh, \kh \TENS \BANG^\ka \kra \TENS \NEG \kh
      }{
        \infer[\rn{init}]{|- \kh, \NEG \kh}{}
        &
        \infer[\BANG]{
          |- \QM^\ka \NEG \kra, \BANG^\ka \kra
        }{
        \infer[\QM]{
          |- \QM^\ka \NEG \kra, \kra
        }{
        \infer[\rn{init}]{
          |- \NEG \kra, \kra
        }{
        }}}
        &
        |- \QM^\infty \Pi,
        \ENC{c} \setminus\set{\smash{\NEG q_1,
          \QM^\ka \D_\kra, \QM^\kb \D_\krb, \QM^\ka \NEG \kra}},
        \NEG \kh
      }}
    \end{gather*}
    There is a symmetric case for contracting the $\kh \TENS \BANG^\kb
    \krb \TENS \NEG \kh$.
    Eventually, the right branch just becomes $|- \QM^\infty \Pi, \NEG
    \kh$, at which point we have:
    \begin{gather*}
      \infer[\rn{contr}, \QM]{
        |- \QM^\infty \Pi, \NEG \kh
      }{
      \infer[\TENS]{
        |- \QM^\infty \Pi, \NEG \kh, \kh \TENS \BANG^\infty \ONE
      }{
        \infer[\rn{init}]{|- \kh, \NEG \kh}{}
        &
        \infer[\BANG]{
          |- \QM^\infty \Pi, \BANG^\infty \ONE
        }{
        \infer*[\rn{weak}]{
          |- \QM^\infty \Pi, \ONE
        }{
          \infer[\ONE]{|- \ONE}{}
        }}
      }}
    \tag*{\qedhere}
    \end{gather*}
  \end{itemize}
\end{proof}

\section{Adequacy of the Encoding via Focusing}
\label{sec:adequacy}

By the contrapositive of theorem~\ref{thm:forward}, if the sequent
$|-_\TWOSI \QM^\infty \Pi, \ENC{c_0}$ is not derivable, then the
\TWORM does not halt from $c_0$.
This gives half of the reduction.
For the converse of theorem~\ref{thm:forward}, we need to show how to
recover a halting trace by searching for proofs of a $\MSEL_\TWOSI$
encoding of a halting problem.
The best way to do this is to build a focused proof which will have
the derived inference rules in the above proof as the only possible
\emph{synthetic} rules, in a sense made precise below.
We will begin by sketching the focused proof system for \SEL that is
sound and complete for the unfocused system of
figure~\ref{fig:sel-rules}, and then show how the synthetic rules for
the encoding are in bijection for all instructions (with a small
correction needed for \k{halt}).


\begin{figure}
  \centering \small%
  \begin{gather*}
    \linfer[\foc{\rn{init}}]{
      |-_\Si \QM^{\vec u} \vec A, \NEG a, \foc{a}
    }{
      (\vec u \in U_\Si)
    }
    \quad
    \linfer[\foc{\TENS}]{
      |-_\Si \QM^{\vec u} \vec A, \W_1, \W_2, \foc{B \TENS C}
    }{
      (\vec u \in U_\Si)
      &
      |-_\Si \QM^{\vec u} \vec A, \W_1, \foc{B}
      &
      |-_\Si \QM^{\vec u} \vec A, \W_2, \foc{C}
    }
    \quad
    \linfer[\foc{\ONE}]{
      |-_\Si \QM^{\vec u} \vec A, \foc{\ONE}
    }{
      (\vec u \in U_\Si)
    }
    \\[1ex]
    \linfer[\foc{\PLUS_1}]{
      |- \W, \foc{A \PLUS B}
    }{
      |- \W, \foc{A}
    }
    \quad
    \linfer[\foc{\PLUS_2}]{
      |- \W, \foc{A \PLUS B}
    }{
      |- \W, \foc{B}
    }
    \quad
    \linfer[\text{no rule for $\ZERO$}]{}{}
    \quad\
    \linfer[\foc{\BANG}]{
      |-_\Si\ \QM^{\vec v} \vec A, \QM^{\vec w} \vec B, \foc{\BANG^u C}
    }{
      \begin{array}[b]{c}
        (u \le_\Si \vec v) \\
        (\vec w \in U_\Si)
      \end{array}
      &
      |-_\Si \QM^{\vec v} \vec A, C
    }
    \quad
    \linfer[\foc{\rn{blur}}]{
      |- \W, \foc{N}
    }{
      |- \W, N
    }
  \end{gather*}
  \hbox to .8\linewidth{\dotfill}
  \vspace{-1em}
  \begin{gather*}
    \linfer[\text{Rules $\PAR$, $\BOT$, $\WITH$, $\TOP$ shared with the unfocused system}]{}{}
    \\[1ex]
    \linfer[\rn{decide}]{
      |- \W, P
    }{
      |- \W, \foc{P}
    }
    \qquad
    \linfer[\rn{ldecide}]{
      |-_\Si \W, \QM^u A
    }{
      (u \notin U_\Si)
      &
      |-_\Si \W, \foc{A}
    }
    \qquad
    \linfer[\rn{udecide}]{
      |-_\Si \W, \QM^u A
    }{
      (u \in U_\Si)
      &
      |-_\Si \W, \QM^u A, \foc{A}
    }
  \end{gather*}
  \vspace{-2em}
  \caption{%
    Inference rules for a focused sequent calculus formulation of
    \SEL.}
  \label{fig:fsel-rules}
\end{figure}

Focusing is a general technique to restrict the non-determinism in a
cut-free sequent proof system.
Though originally defined for classical linear logic
in~\cite{andreoli92jlc}, it is readily extended to many other
logics~\cite{chaudhuri08jar,liang09tcs,nigam09phd}.
This section sketches the basic design of a focused version of the
rules of figure~\ref{fig:sel-rules}, and omits most of the
meta-theoretic proofs of soundness and completeness, for which the
general proof techniques are by now well
known~\cite{chaudhuri08jar,miller07cslb,simmons14tocl}.
To keep things simple, we will define a focused calculus by adding to
the unfocused system a new kind of \emph{focused sequent}, $|- \W,
\foc{A}$, where the formula $A$ is \emph{under focus}.
Contexts written with $\W$, which we call \emph{neutral contexts}, can
contain only positive formulas, atoms, negated atoms, and
$\QM$-formulas.
The rules of the focused proof system for \SEL are depicted in
figure~\ref{fig:fsel-rules}.

Focused sequents are created---reading from conclusion upwards to
premises---from unfocused sequents with neutral contexts by means of
the rules \rn{decide}, \rn{ldecide}, or \rn{udecide}.
In a focused sequent, only the formula under focus can be principal,
and the focus persists on the immediate subformulas of this formula in
the premises, with the exception of the rule $\foc{\BANG}$.
In the base case, for \focrn{init}, the focused atom must find its
negation in the context, while all formulas in the context must be
$\QM$-formulas with unbounded labels.
When the focused formula is negative, the focus is released with the
\focrn{blur} rule, at which point any of the unfocused rules
$\set{\PAR,\BOT,\WITH,\TOP}$ of figure~\ref{fig:sel-rules} can be used
to decompose the formula and its descendants further.
Eventually, when there are no more negative descendants---\ie the
whole context has the form $\W$---a new focused phase is launched
again and the cycle repeats.
Note that the structural rules \rn{contr} and \rn{weak} of the
unfocused calculus are removed in the focused system.
Instead, weakening is folded into \focrn{init}, $\foc{\BANG}$, and
$\foc{\ONE}$, and contraction is folded into $\foc{\TENS}$ and
\rn{udecide}.
The rules \rn{contr} and \rn{weak} remain admissible for either
sequent form in the focused calculus.


\begin{theorem} \label{thm:foc-compl} %
  The \SEL sequent $|- \G$ is provable in the unfocused system of
  figure~\ref{fig:sel-rules} iff it is provable in the focused system
  of figure~\ref{fig:fsel-rules}.
\end{theorem}

\begin{proof}[Sketch]
  Straightforward adaptation of existing proofs of the soundness and
  completeness of focusing, such
  as~\cite{chaudhuri08jar,miller07cslb,simmons14tocl}.
  An instance for \SEL can be found in~\cite[chapter 5]{nigam09phd}.
\end{proof}

\begin{theorem} \label{thm:backward}%
  The \TWORM halts from $c_0$ if $|-_\TWOSI \QM^\infty \Pi, \ENC{c_0}$
  is derivable.
\end{theorem}

\begin{proof}
  We will show instead that the \TWORM halts from $c_0$ if the sequent
  $|-_\TWOSI \QM^\infty \Pi, \ENC{c_0}$ is derivable in the focused
  calculus, and we will moreover extract the halting trace from such a
  focused proof.
  The required result will then follow immediately from
  theorem~\ref{thm:foc-compl}, since any provable \SEL sequent has a
  focused proof.

  Let a focused proof of $|-_\TWOSI \QM^\infty \Pi, \ENC{c}$ (for $c =
  \str{q, v}$) be given.
  We proceed by induction on the lowermost instance of \rn{udecide} in
  this proof.
  Note that the $\MSEL_\TWOSI$ context $\QM^\infty \Pi, \ENC{c}$ is
  neutral; moreover, all the elements of $\ENC{c}$ are either negated
  atoms or $\QM$-prefixed negated atoms with bounded labels.
  So, the only rules of the focusing system that apply to this sequent
  are \rn{ldecide} or \rn{udecide}.
  However, if we use \rn{ldecide}, then the premise becomes
  unprovable, as there is no way to remove an occurrence of $\NEG
  \kra$ or $\NEG \krb$ from a context that also contains $\NEG q$.
  Thus, the only possible rule will be an instance of \rn{udecide},
  with the focused formula in the premise being one of the $\Pi$.
  First, consider the case where the focused formula does not contain
  $\k{h}$, \ie it corresponds to one of the instructions in
  $I\setminus\set{\k{halt}}$.
  In each of these cases, the focused phase that immediately follows
  is deterministic.
  As a characteristic case, suppose the focused formula is $q \TENS
  \BANG^\kb \NEG r$; then we have:
  \begin{gather*}
    \infer[\rn{udecide}]{
      |- \QM^\infty \Pi, \ENC{c}
    }{
    \infer[\foc{\TENS}]{
      |- \QM^\infty \Pi, \ENC{c},
         \foc{\smash{q \TENS \BANG^\kb \NEG r}}
    }{
      \infer[\rn{init}]{|- \NEG q, q}{}
      &
    \infer[\foc{\BANG}]{
      |- \QM^\infty \Pi, \ENC{c}\setminus\set{\NEG q},
         \foc{\smash{\BANG^\kb \NEG r}}
    }{
      |- \QM^\infty \Pi, \ENC{c}\setminus\set{\NEG q},
         \NEG r
    }}}
  \end{gather*}
  The right premise is now itself neutral and an encoding of a
  different configuration.
  We can appeal to the inductive hypothesis to find a halting trace
  for it, to which we can prepend the instruction \k{isza} to get the
  halting trace from $c$.
  A similar argument can be used for the other instructions in
  $I\setminus\set{\k{halt}}$.

  This leaves just the formulas involving $\k{h}$ for the lowermost
  \rn{udecide}.
  We cannot select any formula but $q \TENS \NEG \kh$ from $\Pi$, for
  the derivation would immediately fail because $\kh \notin Q$ and
  there is no $\NEG \kh$ in $\ENC{c}$ to use with \focrn{init}.
  So, as the formula selected is $q \TENS \NEG \kh$, we have:
  \begin{gather*}
    \infer[\rn{udecide}]{
      |- \QM^\infty \Pi, \ENC{c}
    }{
    \infer[\TENS]{
      |- \QM^\infty \Pi, \ENC{c}, \foc{q \TENS \NEG \kh}
    }{
      \infer[\focrn{init}]{|- \NEG q, \foc{q}}{}
      &
      \infer[\focrn{blur}]{
        |- \QM^\infty \Pi, \ENC{c} \setminus\set{\NEG q}, \foc{\NEG \kh}
      }{
        |- \QM^\infty \Pi, \ENC{c} \setminus\set{\NEG q}, \NEG \kh
      }
    }}
  \end{gather*}
  The context of the right premise is now neutral, so the only rule
  that applies to it is \rn{udecide}.
  A simple nested induction will show that sequents of this form $|-
  \QM^\infty \Pi, \ENC{c} \setminus\set{\NEG q}, \NEG \kh$ are always
  derivable in the focused calculus.
  Therefore, the trace that corresponds to the configuration $c$ is
  just the singleton \k{halt}.
\end{proof}

\begin{corollary}
  The derivability of $\MSEL_\TWOSI$ sequents is recursively
  unsolvable.
\end{corollary}

\begin{proof}
  Directly from theorems~\ref{thm:2rm-undec}, \ref{thm:forward}, and
  \ref{thm:backward}.
\end{proof}

\section{Conclusion and Perspectives}
\label{sec:perspectives}

We have given a fairly obvious encoding of a \TWORM in a suitable
instance of \MSEL containing a three element subexponential signature.
The encoding of the \TWORM halting problem is very similar to that
of~\cite{lincoln90apal} for \MAELL; the main difference is in the
encoding of the \k{isz} transitions where we can directly check for
emptiness of the relevant zone instead of making an additive copy of
the world and checking this property in the copy.
Additives are therefore not necessary for undecidability.

Yet, this conclusion is not entirely satisfactory.
If $\MSEL_\TWOSI$ can simulate Turing machines, then it can obviously
simulate a theorem prover that implements a complete search procedure
for \MAELL.
Thus, in an indirect fashion, this paper establishes that additive
behaviour can be encoded using subexponentials and multiplicatives
alone.
It would be interesting to build this encoding of additives more
directly as an embedding of \MAELL---or even just \MALL---in \MSEL.

This work leaves open the questions of decidability of an arbitrary
\MSEL with a two-element signature or a one-element signature; the
latter is equivalent to the decidability of \MELL itself.
We also conjecture that the decision problem for an arbitrary \MSEL
with no unbounded subexponentials is PSPACE-hard, because it is very
likely possible to polynomially and soundly encode a \MALL sequent in
such an \MSEL.

Finally, this undecidability result should be taken as a word of
caution for the increasingly popular uses of \SEL as a logical
framework for the encodings of other systems, such
as~\cite{nigam13concur,nigam11entcs}.
If one is to avoid encoding a decidable problem in terms of an
undecidable one, subexponentials must be used very carefully.


\smallskip \noindent%
\textbf{Historical note}: The undecidability result presented here is
from an unpublished paper from 2009, cited as the source of the result
in Nigam's Ph.D. thesis from the same year~\cite[p. 103]{nigam09phd}.
Nigam has also published an indirect proof in~\cite{nigam12lics},
using the same strategy and roughly the same encoding, but this
version also uses the additive unit $\TOP$ for halting states and is
therefore not strictly in \MSEL.

\end{document}